\documentclass[12pt,a4paper]{article}

\usepackage{ulem}
\usepackage{color,amsmath}
\usepackage{amsfonts}
\usepackage{amssymb}
\usepackage{graphicx}
\usepackage{epsfig}
\usepackage{enumerate}
\usepackage{stmaryrd}
\usepackage{varioref,exscale,relsize}
\usepackage{amsthm}
\usepackage{mathrsfs}
\usepackage{amsfonts}
\IfFileExists{url.sty}{\usepackage{url}}
                     {\newcommand{\url}{\text}}

\newtheorem{proposition}{Proposition}

\DeclareRobustCommand{\FIN}{%
    \ifmmode 
    \else \leavevmode\unskip\penalty9999 \hbox{}\nobreak\hfill
    \fi
    $\bullet$ \vspace{5mm}}

\begin{document}

\begin{center}
{\Large \textbf{Resistant estimates for high dimensional and
functional data based on random projections}}


{\large Ricardo Fraiman $^{\ast}$, and Marcela Svarc $^{\ast\ast}${\footnote{Corresponding author:
Marcela Svarc, Departamento de Matem\'aticas y Ciencias, Universidad  de San Andr\'es, Vito Dumas 284, Victoria (1644), Buenos Aires, Argentina. Email: msvarc@udesa.edu.ar}}}

\

\noindent \emph{$^{\ast}$Universidad de San Andr\'es, Argentina and Universidad de la Rep\'ublica, Uruguay.}
\\[0pt]\emph{$^{\ast\ast}$ Universidad de San Andr\'es and CONICET, Argentina. }
\
%

\end{center}

%

\

\begin{abstract}
We herein propose a new robust estimation method based on random projections that is adaptive and, automatically produces a robust estimate, while enabling easy computations for high or infinite dimensional data. Under some restricted contamination models, the procedure is robust and attains full efficiency. We tested the method using both simulated and real data.
\newline \noindent\textbf{Keywords:}
Robust Estimates, Location and Scatter Estimates, Trimming Estimates.
\newline \noindent\textbf{Running Title:}
Random Projection Robust Estimates.
\end{abstract}

\section{Introduction}

In the problem of robust estimation in  high dimensional and functional data, a number of different issues arise that cause the well known robust multivariate estimators to fail, and new ideas are required to address these.

Before stating our proposal, we first review the origins
of robustness in order to identify some of the main ideas necessary to
guide us in this complex problem.

Ronchetti (2010) who stated that ``the primary goal of robust statistics is
the development of procedures which are still reliable and
reasonably efficient under small deviations from the model, i.e.
when the underlying distribution lies in a neighborhood of the
assumed model. Robust statistics is then an extension of
parametric statistics, taking into account that parametric models
are at best only approximations to reality. The field is now some
50 years old. Indeed one can consider Tukey (1960), Huber (1964),
and Hampel (1968) the fundamental papers which laid the
foundations of modern robust statistics."

\bf The median. \rm For one dimensional data, the oldest robust
estimates are the median and the $\alpha$--trimmed means. The
median is in several respects the `most robust' consistent
possible location estimate, and great efforts have been made to extend the median and the trimmed means to higher dimensions.
The extension of these two simple notions to multivariate data  took place via the two main avenues of impartial trimmed means and the concepts of data depth.
For the multivariate case Gordaliza (1991) introduced the procedure of impartial trimmed means, which was then
extended by Cuesta and Fraiman (2006)  for functional data to obtain resistant estimates of the center of a functional distribution.
More recently, several proposals considering trimming procedures have been introduced, among them Dolia et al (2007) and Cuevas et al (2008).
A number of depth definitions have been introduced for multivariate data, the best known being the
 Tukey Depth (Tukey, 1975) and  the Simplicial Depth (Liu, 1988) and an interesting review of this subject is provided by Serfling (2006).
 For the case of functional data, several proposals have recently been made, including
 Fraiman and Muniz (2001),  Cuevas and Fraiman (2009), Gervini (2010), among others.

\bf Diagnostic and automatic robust procedures. \rm Two
complementary approaches were developed to achieve robustness.
Broadly speaking, the different procedures used  to flag outliers are known as diagnostic procedures, while the term `robustness' refers to statistical procedures that are insensitive to the effect of outliers. As observed by Serfling (2006), for each depth function one can find a function associated with outlyingness.
Hence, several of the foregoing procedures can help in the detection of outliers, such as the proposals of Fraiman and Muniz (2001) and Gervini (2010).  Febrero et al (2008) introduced a depth measure for functional data that they used for the identification of outliers.

Many procedures for the identification of outliers have been introduced, among them Rousseeuw and  Van Zomeren (1990), Pe\~{n}a and Prieto (2001) and
 Febrero et al (2007).

In robustness, a class of estimators  can be  defined in general
in terms of the optimization of a specified target functional for the
population version of the estimate. Maronna et al (2006), in Chapter 6, presented a
detailed review of the different estimates of location and scatter
 for the multivariate case, describing the properties of each.

In diagnostics, the  problem lies in the identification of sub--samples whose
 deletion maximally changes a statistic of interest as measured by an appropriate target
function.

The diagnostic procedure straightforward for one dimensional data, but becomes much
more difficult for multivariate data. This is also the case for
automatic robust procedures.


\bf Adaptive estimates. \rm In the early 1970's, some adaptive
estimates were proposed for the location model. The word adaptive
implied that the estimate adjusted itself in some way to the data concerned.
The simplest idea was to consider an M-estimate using Huber's
$\psi_K$ function, where the parameter $K$ (that compromises
robustness and efficiency) depends on the data, i.e. $K=K(X_1,
\ldots, X_n)$ (see for instance Yohai, 1974). Roughly speaking, if
there are no outliers on the sample, the estimate would tend to
choose large values of $K$, while choosing smaller values of $K$
correspond to samples that have outliers. These interesting ideas
 have been neglected ever since.

\bf Neighborhoods. \rm The question arises whether full neighborhoods
should be used for some distances between probabilities, or contamination neighborhood should. As mentioned above, we
expect robust procedures to be stable under a small
deviation from the model, i.e., when the underlying true
distribution $P$ lies in the neighborhood of the assumed model
$P_0$. Some descriptions of neighborhoods are therefore required. We consider two types, the first being, $
\mathcal P_{\epsilon, D} = \{P: D(P,P_0) < \epsilon\},$ where $D$ is an appropriate distance between probabilities, derived from, for example Prohorov, Kolgomorov or Wasserstein
metrics. The second type is of the form $
\mathcal P_{\epsilon} = \{P: P = (1-\epsilon) P_0 + \epsilon Q\},$ where $Q$ is an arbitrary probability measure. Typically,
$\mathcal P_{\epsilon} \subset \mathcal P_{\epsilon, D}$, which is
the case for Prohorov or Kolmogorov (total variation)
distances, for example, $\mathcal P_{\epsilon}$ is particularly appealing, however.
It may be perceived that with probability $(1-\epsilon)$, an
observation of the sample comes from a distribution $P_0$, while
for probability $\epsilon$ it comes from $Q$.
Recently, Alqallaf et al (2009) considered non standard data contamination
models, such as componentwise outliers.

Several notions and measures of robustness have been defined to
provide a detailed description of the robustness characteristics of an
estimate.
 The main ideas are:
\begin{itemize}

\item
Solving minimax problems in a neighborhood of the ``true" model.
The first approach to formalize the robustness problem was that of Huber`s
(1964, 1981) minimax theory, where the statistical problem is
viewed as a game between the Nature (which chooses a distribution
in the neighborhood of the model) and the statistician (who
chooses a statistical procedure in a given class). The
statistician achieves robustness by constructing a minimax
procedure that minimizes a lost criterion in which the worst
possible distribution in the neighborhood is considered.

\item
Qualitative robustness. The notion of qualitative robustness was
introduced by Hampel (1968) in terms of the equicontinuity of the distribution
of the estimate with respect to the Prohorov distance between
probabilities.

\item Measures of robustness. We herein briefly describe the most important measures of robustness. The  influence curve is
the infinitesimal approach was introduced by Hampel (1968) within the
framework of robust estimation.
The finite sample breakdown point (introduced by Donoho and Huber,
1983) of an estimate is the maximum number of observations that
can be replaced arbitrarily while the estimate remains bounded
(or far away from the boundary of the parameter space when the
parameter space is compact).
The maximal asymptotic bias as a function of the contamination fraction, measures the worst behavior of an
estimate under a fixed proportion of contamination.

\end{itemize}

\bf Equivariance properties. \rm We expect that any reasonable
estimate \it would only depend on the data, and not on the
coordinates axis  nor on the scale considered. \rm  In the one
dimensional problem, the equivariance properties are just the location and the
scale equivariance. In the multivariate case, this corresponds to
the location equivariance, the orthogonal equivariance, and the global scale
equivariance. These properties are still sensible in the infinite
dimensional case. However, in the finite
dimensional case, the affine equivariance properties often also required.
 However, this last condition is far from mandatory, and could also be a bad idea. For instance, for
contours of multivariate quantiles, the requirement of affine equivariance
implies that the contours will be convex sets, losing the important
features shown by the non--convex versions that are orthogonal,
translation and scale equivariant but not affine equivariant.
 We consider that the affine equivariance requirement is inadequate
 for both high and infinite
dimensional data. Alqallaf et al (2009) proved that standard
high-breakdown affine equivariant estimators propagate outliers
under componentwise contamination when the dimension $d$ is high.
There are many examples of estimates of location and scatter that are not affine equivariant
but have many appealing properties (see for instance Maronna and Zamar
(2002) or Fraiman and Pateiro-Lopez (2011)).

\bf Computational burden. \rm Most of the robust estimates for
multivariate data, particularly  those that have improved properties
of robustness, are very expensive from a computational point of view
and, their computation is unfeasible for high dimensional data.
This is also the case with some of the data depth based
estimates. In the multivariate case, several
attempts have been made in recent years to overcome these problems, by considering
resampling algorithms on the initial candidates that
 allow a substantial reduction in the number of
  candidates required to obtain a good approximation to the optimal solution
(see for instance Van Aelst and Rousseeuw, 2009 or Salibian and Yohai, 2006).
In the case of high dimensional or functional data, it is essential that the estimates can be
computed in a reasonable time.

\bf Random projections. \rm The idea of using random projections
has been already used by several authors in many different
contexts. An appealing theoretical framework for the use of random
projections can be found in  Cuesta, Fraiman and Ransford (2006,
2007).
\

We herein propose a new robust estimation procedure based on
random projections that is data--adaptive, produces automatic robust
estimates and combines ideas from the diagnostic framework. In addition it
 has very low cost when applied to high or infinite dimensional data. The method yields
robust estimates that are location and scale equivariant and are also invariant under
unitary operators (orthogonal transformations in the finite dimensional case).

The remainder of the paper is organized as follows. In Section 2
we present our general method, and we exemplify with location and
scatter estimates. In Section 3 we provide some asymptotic results.
Section 4 is devoted to analysis of the performance of the procedure on
simulations of the finite dimensional case and also for the case of
functional data. In Section 5 we analyze an example using  real data.
Section 6 includes some final remarks.

\section{Our setup}

We start by fixing some notation. Let $\mathbb{H}$ be a real
separable Hilbert space, with inner product $<\cdot,\cdot>$, and the
induced norm  $\Vert \cdot \Vert_{\mathbb{H}} =: \Vert
\cdot\Vert$. $X$ denote a random element in $\mathbb{H}$ with
distribution $P$ in a contamination neighborhood of $P_0$, our
``target distribution".

The data are given by a random sample $X_1, \ldots, X_n$ of iid
random elements with the same distribution as $X$.

We are interested in estimating different parameters, including the
mean, the covariance operator, and the principal components of a
random element $X_0$ with distribution $P_0$ based on the sample
$X_1, \ldots, X_n$. The problem lies in the fact that the distribution of the data is given by $P$ rather than $P_0$.

We assume that $\mathbb E(X_0)=\mu_0$ is finite and that the
covariance operator $\Sigma_{X_0}$, is positive definite when
$\mathbb{H}$ is finite dimensional and compact and self-adjoint if
$\mathbb{H}$ is infinite dimensional.

Our aim is to introduce a general method of obtaining consistent
robust estimates in this framework that is both suitable for different problems
and computationally feasible.

 If we were able to select the
 sub--sample $X_{i_1}, \ldots, X_{i_m}$ of those random elements that
 have distribution $P_0$,  from the sample $X_1, \ldots, X_n$ ,
 the problem becomes trivial. We may apply
 usual ``non--robust" estimates to the ``good data".

We propose a trimming scheme based on random projections to
select a sub--sample $X_{i_1}, \ldots, X_{i_m}$, which we hereafter
denote the RT--procedure. The maximum number of observations to
 be trimmed is $\alpha$ $(0<\alpha<0.5)$. The goal ``is to
discard only those observations that are far apart from the bulk of the
data". Our method is adaptive in that it stops
searching for outliers in a data dependent way.


In the presence of outliers that lie far away from the
bulk of the data, there are always several directions $u, \Vert u \Vert
=1$ (a set of positive probability on the unit sphere with respect
to the induced probability measure) where on the projected data,
$\{<X_1, u>, \ldots,<X_n, u>\}$, the outliers are  distant from
the majority of the observations. If an observation is an outlier
this phenomenon is observed in many directions. Therefore, the
maximal gap of the one dimensional projected data, will be large,
or more precisely, larger than expected. We recall that the maximal
gap is the maximum distance between two contiguous observations. On
the other hand, those observations that are deep in the bulk of
the data set also have deep one dimensional projections in every
direction.

 However, we may have a shadow cone, where the projected outliers
 are ``hidden" for some directions $u$, as shown on Figure \ref{shadowcone}.

We seek to define a pruning criterion that combines these two
facts. We consider two parameters,  the first being $\alpha$,
$\alpha \in [0,0.5]$,  the maximal trimming rate, the second being,
$\textit{maxiter}$, the maximal number of random directions that we
 generate before deciding that all the outliers have been removed.

\begin{figure}[h]
\begin{center}
\includegraphics[width=6cm]{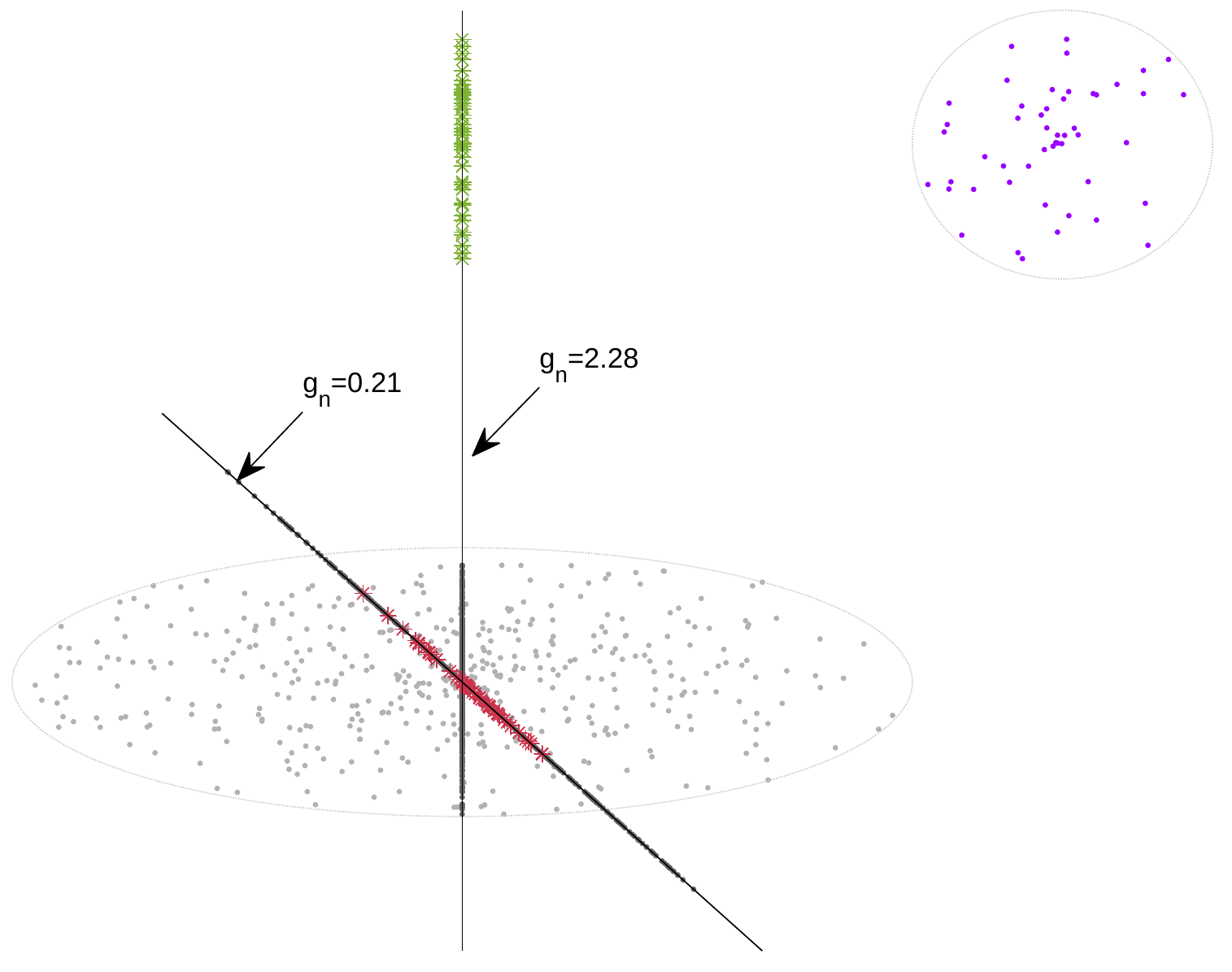}
\end{center}
\vspace{-15pt} \caption{The projections on the direction of the y
axis are useful to detect the outliers, the projections on the
direction of $y=-x$ are on the shadow cone.} \label{shadowcone}
\end{figure}

\bf The subsample selection algorithm. \rm

The following procedure should be repeated either $\lfloor n
\alpha \rfloor $ times (where $\lfloor k\rfloor$ represents the
largest integer smaller than or equal to $k$) or until the number of
random directions generated, $\textit{numdir}$, reaches
$\textit{maxiter}$, i.e. while $\textit{numdir} \leq
\textit{maxiter}$.

\begin{itemize}
\item Set $\textit{numdir}=1$.
\item Generate a direction $h$ at random according to a non-degenerate gaussian measure $\nu$ .
\item Project the data on that direction, $Y_{i}^{h}=<X_i,h>$.
\item Sort the projected observations in ascending order $Y_{(1)}^{h}  \leq \dots \leq Y_{(n)}^{h}$.
\item Consider the maximum gap $g_n=\max_{i=1,\dots,n-1}{|Y_{(i+1)}^{h}-Y_{(i)}^{h}|} $.
\item If $g_n\geq c_d$, $\dag$ trim the observation whose one--dimensional projection is most distant from the median of the projected observations,
  i.e. $X_i$ such that $i= arg \max \{|Y^{h}_{i}-\widetilde{Y^h}|$ for $i=1,\dots,n\}$  where $\widetilde{Y^h}=median\left( Y_{\left( 1\right) }^{h},\ldots
,Y_{\left( n\right) }^{h}\right) $.
\item Else if $g_n < c_d$, set $\textit{numdir}=\textit{numdir}+1$.
\end{itemize}

We denote to be $\gamma$ the effective trimming rate, i.e., the actual rate of observations that have been reduced.

The value of the threshold $c_d$ on $\dag$ should be related to the limiting behavior of the largest spacing.
Deheuvels (1984), characterized it in terms of the local behavior of the density of the observations
(in our case the one dimensional projections) in the neighborhood of the point at which
it reaches its minimum value. According to Theorem 4, he proved that if the distribution has compact support and
the density is bounded away from zero, i.e. there exits $x_0$ such that $f(x)>f(x_0)>0$ for
every $x_0 \in Support(f)$, and if $f$ has first derivative then, $
\liminf_{n\rightarrow\infty}\frac{ng_n f(x_0) -\log(n)}{\log(\log(n))}=-1.
$

This result suggest to take (for sufficiently large values of $n$)
$
c_d=k\frac{\log(n)-\log(\log(n))}{nf(x_0)},
$
where $k$ is a suitable constant greater than 1.
There are many univariate distributions that do not have compact
support, among which the normal distribution is probably the best known. In practice all the
observable data sets
 have bounded support, hence our consideration of this
expression might enable us to determine a useful threshold.

\noindent { \bf Remark 1.} If $\mathbb{H}$ is a finite dimensional
space the random directions can be generated according to an
uniform distribution on the unit sphere.

Finally, we compute the ordinary (non--robust estimate) of our target
parameters. For instance, the ML mean and covariance estimates
are computed only with those observations that have not been trimmed.
 This implies, that if we assign weight $w(X_i)=w_i=0$ to those observations
 that have been pruned and $w(X_i)=w_i=1$ to the remainder of the observations,
  then the mean and the covariance function estimates
 are given by

\begin{equation*}
\widehat{\mu }\left( t\right)
=\frac{\sum_{i=1}^{n}w_{i}X_{i}\left( t\right)
}{\sum_{i=1}^{n}w_{i}},
\end{equation*}

\begin{equation*}
\widehat{\Sigma }\left( s,t\right) =\frac{\sum_{i=1}^{n}w_{i}\left(
X_{i}\left( s\right) -\widehat{\mu }\right) \left( X_{i}\left( t\right) -%
\widehat{\mu }\right) }{\sum_{i=1}^{n}w_{i}},
\end{equation*}

respectively, where $\sum_{i=1}^{n}w_{i}=n-\lfloor \gamma n
\rfloor$.

\bigskip

\textbf{Equivariance properties}

\bigskip
The estimates are equivariant under translations, re--scaling and
orthogonal transformations. Indeed, it follows from the fact that
the weights assigned to the observations are not affected by these
operations, because we know that the usual mean and covariance
estimates are not affected by them.

It is clear that a shift and re--scaling with a positive
constant will not affect the order statistics of the
 projections, and so the trimmed observations
 will be the same.
 Also, if \emph{T} is an orthogonal transformation, then $\left\langle X_{i},h\right\rangle =\left\langle \emph{T}X_{i},\emph{T}%
h\right\rangle $ (because it is an isometry), which means that as
the projections do not change, so neither do the weights.

\section{Some comments on the asymptotic behavior}
 Let $X_1,\ldots, X_n$ be iid random elements with the same distribution as
$X$ defined on $(\Omega, \mathcal A)$, and $h_1, \ldots, h_k$ iid
random directions on the unit sphere, independent of $X_1, \ldots,
X_n$ and defined on $(\Omega', \mathcal A')$. After performing the
algorithm we obtain a subset of random elements $X_{i_1}, \ldots,
X_{i_m}$, where $m \leq n$ is random, and each
$X_{i_j}(\omega,\omega'), j=1, \ldots, m$ depends on the set of
random elements $\{X_1, \ldots, X_n, h_1, \ldots, h_k\}$.

The distribution of $X$ is $P = (1-\epsilon) P_0 + \epsilon Q$,
where $P_0$ is the core distribution, $Q$ is an arbitrary
contamination and $\epsilon \in [0,0.5)$.

Let $P_m(\omega,\omega')$ be the empirical measure associated with
the set of random elements $\{X_{i_1}, \ldots, X_{i_m}\}.$

We wish to identify the conditions under which $P_m$
 converges weakly towards $P_0$, which would imply that
for any continuous functional $T(P)$, the estimate defined as
$T(P_m)$ will converge  to $T(P_0)$, i.e., we obtain almost sure
consistency without assuming any kind of symmetry.

Let $X_0$ be a random element with distribution $P_0$. As
mentioned above we assume that $\mathbb E(\Vert X_0\Vert^2)$
is finite and the covariance operator $\Sigma_{X_0}$, is positive
definite when $\mathbb{H}$ is finite dimensional and compact and
self-adjoint if $\mathbb{H}$ is infinite dimensional. Without any loss
of generality we also assume that $\mathbb E (X_0) = 0$.

We now introduce  the following hypotheses:

\bf A1. \rm $P = (1-\epsilon)P_0 + \epsilon Q$, for some $0 \leq
\epsilon < 0.5$.

\bf A2.- \rm $X_0$ has a bounded support $S_0$, and $Q$ is such
that its support $S_Q$ verifies that $d(S_0, S_Q) \geq \delta
>0$.

This means that the outliers are far from the bulk of the data.

\bf A3.-\rm  For all $h, \Vert h \Vert = 1$, $X_0(h):= <X_0, h>$
has a derivable density $f_h$ that fulfils

\begin{equation}
\label{A3}
 \inf_{t \in S_h} \vert f_h(t) \vert  > \eta,
\end{equation}
for some $\eta > 0$, where $S_h$ represents for the support of $f_h$.

Let $\{X_{j_1}, \ldots, X_{j_r}\}$ be the subset of $\{X_1, \ldots,
X_n\}$ with $P_0$, $P_r$ being the corresponding empirical
distribution, and  $I_n = \{j_1, \ldots, j_r\}$.

Because of \bf A1, \rm we may assume that
$
X_i = Z_i X_{i0} + (1-Z_i) X_{iQ}, \ \ i=1, \ldots, n,
$
where $X_{10}, \ldots, X_{n0}$ are iid random elements with
distribution $P_0$, $X_{1Q}, \ldots, X_{nQ}$ are iid random
elements with distribution $Q$ and $Z_1, \ldots, Z_n$ are iid
random variables with distribution Bernoulli($\epsilon$), being
independent the three sets of random elements.

The aim is to show that $d_{\Pi} (P_m, P_0) \to 0$ a.s., where
$d_{\Pi}$ stands for the Prohorov distance when $n \to \infty$,
and $k=k(n) \to \infty$ at an appropriate rate. This would imply
that for any continuous functional $T(P)$, the estimate defined as
$T(P_m)$ would converge  to $T(P_0)$. Moreover, we show that in the
finite dimensional case, the final estimates will attain
``full--efficiency", meaning that the output of the algorithm will
be  the set of random elements with distribution $P_0$. In the
infinite dimensional case, an additional assumption will be
required in order to attain the ``full-efficiency". Otherwise,
under some symmetry conditions on the distribution $P_0$,
regardless of the distribution $Q$, the estimate will still be
consistent.

A sketch of the proof may be made on the following lines:


 First we prove that since $\epsilon < 0.5$, for a sufficiently large $n$ we
 have that $r > n-r$ a.s.

Since $r = \sum_{i=1}^n Z_i$, this is equivalent to show that
$\sum_{i=1}^n Z_i > n - \sum_{i=1}^n Z_i$ a.s. for sufficiently large $n$,
 which follows because by the Hoeffding inequality we have that
$
P(\sum_{i=1}^n Z_i > \frac{n}{2}) \leq e^{-2n(0.5 -\epsilon)^2},
$
and $\epsilon < 0.5$.

Second, we show that given $X_i$ with distribution $P_0$ and any $X_j$ with distribution $Q$,
 for $\omega$ in a subset of probability $1$ on
$\Omega$, there exists $\delta_1(\omega)
> 0$ and a set of strictly positive probability of directions on $\Omega'$ for
which

\begin{equation}
\label{lejos}
 \vert <X_i, h> - <X_j, h> \vert > \delta_1(\omega).
\end{equation}

Indeed, by assumption \bf A2 \rm , we have that $< X_i - X_j, X_i
- X_j >
> \delta$ for all $j$ such that $X_j$ has distribution $Q$, with
probability $1$. There is therefore a cone of directions $h$ where
(\ref{lejos}) holds.

This implies that in the presence of outliers, for sufficiently large values of $k$
we can find as many directions as are required for which the
maximal spacing of the projected data is be larger than
$\delta_1$.

Recall that the value of the threshold $c_d$ on $\dag$ is
related to the limiting behavior of the largest spacing. Deheuvels
(1984) characterized it in terms of the local behavior of the
density of the observations (in our case the one dimensional
projections) in a neighborhood of the point where it attains its
minimum spacing. On Theorem 4, he proves that if the distribution has
compact support and the density is bounded away from zero, i.e.
there exits $x_0$ such that $f(x)\geq f(x_0)>0$ for every $x \in
Support(f)$, and if $f$ has first derivative then
$$
-1= \liminf_{n\rightarrow\infty}\frac{ng_n f(x_0)
-\log(n)}{\log(\log(n))} < \limsup_{n\rightarrow\infty}\frac{ng_n
f(x_0) -\log(n)}{\log(\log(n))}= 1, \ \ \mbox{a.s.},
$$

thus, the maximal spacing $g_n$ behaves like
$\frac{\log(n)\pm\log(\log(n))}{nf(x_0)}$, which converges to zero
sufficiently quickly.

Our threshold, which is a constant, is defined as,
$
c_d=k\frac{\log(n)-\log(\log(n))}{nf(x_0)},
$
where $k$ is a suitable constant greater than $1$.

Since $r > n-r$, the median of the projected data must lie
between $\min_{i \in I_n} <X_i, h>$ and $\max_{i \in I_n} <X_i,
h>$.

The foregoing statements and assumption \bf A3 \rm  imply that for
directions $h$ fulfilling (\ref{lejos}) the maximal spacing will
be achieved between some $<X_i, h>$ and $<X_j, h>$ for $X_i$ with
distribution $P_0$ and $X_j$ with distribution $Q$ or for both
$X_i, X_j$ with distribution $Q$. In both cases an observation
with distribution $Q$ will be deleted. Therefore, for sufficiently large values of
 $k$, all outliers will be deleted.


 The foregoing proofs are insufficient, however. We already know from the previous
assertion that for large enough values of $k$ the set
$\{X_{i_1}, \ldots, X_{i_m}\} \subset \{X_{j_1}, \ldots, X_{j_r}\},$
but we must still prove that the set $X_{i_1}, \ldots, X_{i_m}$
``behaves as an iid sample of $P_0$", and that $m \to \infty$. To
achieve this goal we need conditions more stringent than (\ref{A3}) to
hold (see, (\ref{masfuerte}) below).

For each fixed direction $h_1$, let $g_r(h_1)$ denote the
maximal spacing of the set $\{ <X_j, h_1>: j \in I_n \}$, with
cardinal $r$. From assumption \bf A3, \rm and the fact that $c_d
> \frac{log (r) + log(log (r))}{r f(x_0)}$ we have that $
P( g_r(h_1) > c_d \ \ \mbox{i.o.} ) = 0,$
and thus, $g_r(h_1) < c_d$ for $r \geq r_1(\omega)$ for almost all
$\omega \in \Omega$.

If that holds, then
\begin{equation}
\label{masfuerte}
 P( \max_{1 \leq l \leq k} g_r(h_l) > c_d  \ \
\mbox{i.o.} ) = 0,
\end{equation}
for a sequence $k=k(n) \to \infty$, the procedure will end with
exactly all observations on the set $\{X_{j_1}, \ldots,
X_{j_r}\}$, and the estimate will attain full efficiency.

In the finite dimensional case, (\ref{masfuerte}) is a consequence
of assumption \bf A3 \rm (see Proposition \ref{coupling} below),
but this is not the case in the infinite dimensional setting.

Otherwise, we would need to require some symmetry assumptions of the pair $(T, P_0)$.
 More precisely, we would require that
\begin{itemize}
\item $X_0$ has a symmetric distribution around $\mu := \mathbb
E(X_0)$, i.e. $X_0 - \mu$ and $\mu-X_0$ have the same distribution,
and that $T(P_0) = T(\tau_K(P_0))$, where $\tau_K(P_0)$ is the
distribution of a truncation of the random variable $X_0$, i.e.
$\tau_K(X_0) = X_0 \mathbb I_{B(\mu,K)}$.
\item The density $f_h(t)$ is a decreasing function around $<\mu, h>$ for all
$h$,
\end{itemize}
to obtain consistency. Indeed, it is be a consequence of the
symmetry and Theorem 4 in Deheuvels (1984), where it is shown that
for any $\delta > 0$, there exists almost surely a value of $N$ such that,
for any $n \geq N$, the maximal spacing interval is included in
$(x_0 - \delta, x_0 + \delta)$.


\begin{proposition}
\label{coupling} Assume that $S_0 \subset \mathbb R^d$ the support
of $P_0$ is a compact connected set with Lebesgue boundary set
null, i.e. $\vert\partial S_0 \vert = 0$, and that there exists
$\eta > 0$ such that the density of $X_0$ fulfills $f > \eta$ on
$S_0$. Given a sample $\{X_1, \ldots, X_n\}$ of iid random vectors
with distribution $P_0$, let
$$
\Delta_n(P_0) = \sup_{r \in \mathbb R} \left\{ \exists \ x \in
\mathbb R^d \ \ \mbox{with}  \ \ x + r B(0,1) \subset S_0
\setminus \{X_1, \ldots, X_n\}\right\},
$$
be the maximal multivariate spacing. We then have that
\begin{itemize}
\item [i)] There exists a compact set $ S_1$, $S_0 \subset
S_1 \subset \mathbb R^d$, $\vert S_1 \vert = 1/\eta$, $\vert
\partial S_1 \vert =0$ such that
$$
P(\Delta_n(P_0) > c) \leq P(\Delta_n(P_1) > c), \ \ \forall c >0,
$$
where $\Delta_n(P_1)$ is the maximal spacing corresponding to a
uniform random sample $U_1, \ldots, U_n$ on $S_1$.
\item [ii)]
$
P( \max_{1 \leq l \leq k} g_r(h_l) > c) \leq P(\Delta_n(P_1) > c),
$
for all $c > 0$.
\item [iii)]
$$
n V_n - log n - (d-1) loglog n - log \alpha  \
{\rightarrow}^\omega \ U,
$$
where $V_n = \Delta_n(P_1)^d  \ \frac{\vert B(0,1)\vert}{\vert S_1
\vert}$, $U$ has the extreme value distribution $P(U \leq u) =
e^{-e^{-u}}$ and $\alpha$ and is a known constant. Moreover,
$$
d-1 = \liminf_{n \to \infty} \frac{ n V_n - log n}{loglog n} \leq
\limsup_{n \to \infty} \frac{ n V_n - log n}{loglog n} = d-1, \ \
\mbox{a.s.}
$$

\end{itemize}
\end{proposition}

\begin{proof} The proof of i) is based on a coupling argument. We
first demonstrate it for the simple case of $d=1$. It may be assumed
 without loss of generality that $S_0 =
[0,1]$, and $f(x)> \eta \ \ \forall \ x \in [0,1]$. Let $g$ be the
uniform distribution on $S_1 = [0, \frac{1}{\eta}]$, and $U_1,
\ldots, U_n$ iid random variables uniformly distributed on
$[0,1]$. Let $F$ and $G$ stand for the cumulative distribution
functions, and define $X_i = F^{-1}(U_i), \ \ Y_i = G^{-1}(U_i)$,
iid samples with distribution $F$ and $G$ respectively. Since for
any pair $i,j$, $F(X_i)=G(X_i), \ \ F(X_j)=G(X_j)$, and the
derivative of $F$ is uniformly greater than that of $G$, we have
that $\vert X_i -  X_j \vert < \vert Y_i - Y_j \vert,$
which concludes the proof.


The argument for $d > 1$ is a little more involved and can be found in  Ferrari et al. (2011).


\


ii) Let $X^{(1)}(x)$ be the nearest neighbor of $x$ among $X_1,
\ldots, X_n$. For any $h$, $\Vert h \Vert =1$, and for any pair $i,j$
we have that $ \vert <h,X_i> - < h, X_j> \vert \leq \Vert X_i -
X_j \Vert$. Therefore,
\begin{equation}
\label{desigualdad}
 \max_{1 \leq l \leq k} g_n(h_l) \leq \Delta_n(P_0),
\end{equation}
and the conclusion follows from (\ref{desigualdad}) and  i).

 iii) These results were obtained by Janson (1987, Theorem1).

\end{proof}

\section{Simulations}
In this Section we discuss the performance of the location,
correlation estimates proposed herein, as well as robust
principal components estimates. In Subsection \ref{finito} we
analyze the finite dimensional case and the infinite dimensional
case is considered in Subsection \ref{infinito}.  In both cases we
compare the RT--estimate with other classical robust estimates and
with a benchmark, which is defined below.
\bigskip
\subsection{Finite Dimensional Case}
\label{finito}

 We report on the results of a Monte Carlo study, the aim of which was
to compare the performance of the Random Trimming (RT) estimates
for location and the correlation matrix. It is well known, that
the distributions of the original data set $P$ and the trimmed
data set $\widetilde{P}$ are different, and if $P$ is elliptical the
center $\mathbf{\mu}$, coincides with the center
$\widetilde{\mathbf{\mu}}$ of $\widetilde{P}$. However, the
covariance estimate $\Sigma$ of $P$ and the covariance estimate
$\widetilde{\Sigma}$ of $\widetilde{P}$ have the same shape but
differ by a constant $c$. In order to avoid this problem, we
estimated the correlation matrix for every case that did not
depend on constants. Another approach could have been the computation the
minimum and maximum eigenvalues or the matrix condition number. We herein
compare our estimates with two different estimates and a
benchmark, which we describe below.

a) The Minimum Volume Ellipsoid (MVE) estimate. The MVE estimate
was introduced by Rousseeuw (1987); among all the ellipsoids  $\{x
: d(x,\mu, \Sigma)<c\}$ that contain at least half the data
points, the one given by the MVE estimate has the minimum volume.
Since the computation of this estimate numerically is very expensive unless
$p$ and $n$ are small, a sub--sampling procedure as
described by Maronna et al. (2006) is carried out, we denote these
estimates (MVEF). We consider these estimates because, despite
being very inefficient, they achieve the highest possible
breakdown point, if the points are in a general position (this means
that no hyperplane contains more than $p$ points). In addition,
these estimates are the initial estimates usually considered for the
computation of S-estimates, and a number of statistical programs,
including S-plus, R and SAS, have internal routines to compute them.
In these cases the covariance matrix estimation also differs from
that of the population by constant, then we compute the correlation
matrix.

b) As a benchmark, we compute the mean and correlation classical
estimates from the data sets without outliers, and we denote these
\emph{trick} estimates. It is clear that this benchmark attains
the best possible behavior.

c) Recently, Gervini (2010) introduced a trimmed estimate for
location and scatter, mainly with a view to handling functional data
but it was also defined for general Hilbert spaces. Again the idea was
to trim the observations that are far away from the bulk of the
data set. For each observation $X_i$ he defines the $\alpha$--
radius, $r_i$, $\alpha  \in [0,0.5] $, to be the radius of the
smallest ball centered on $X_i$ that contains $100 \alpha \%$ of
the observations. In the second stage he trimmed the $100\beta \%$, $\beta \in
[0,0.5] $, of the observations with higher $r_i$. We
denote these Inter--distances Trimmed Estimates (IT).

Figure \ref{corona} shows an example where the performances of the  RT and  IT estimates are very different.
The data are displayed in Figure \ref{corona} (a); the grey points correspond to the central distribution and the
green ones are the outliers. The contamination rate is $20\%$. In Figures \ref{corona} (b) and (d) the
red dots correspond to the observations trimmed by considering the IT procedure where
the trimming rates were $20$ and $40\%$ respectively. In the first case it was only the outliers that were trimmed,
while in the second case some observations from the core distribution were also pruned. In Figures \ref{corona} (c) and (d) the
pink dots correspond to the observations trimmed by the RT procedure for trimming rates of $20$ and $40\%$
respectively. In both cases only the outliers were trimmed.

 \begin{figure}[h]
\begin{center}
\includegraphics[width=6cm]{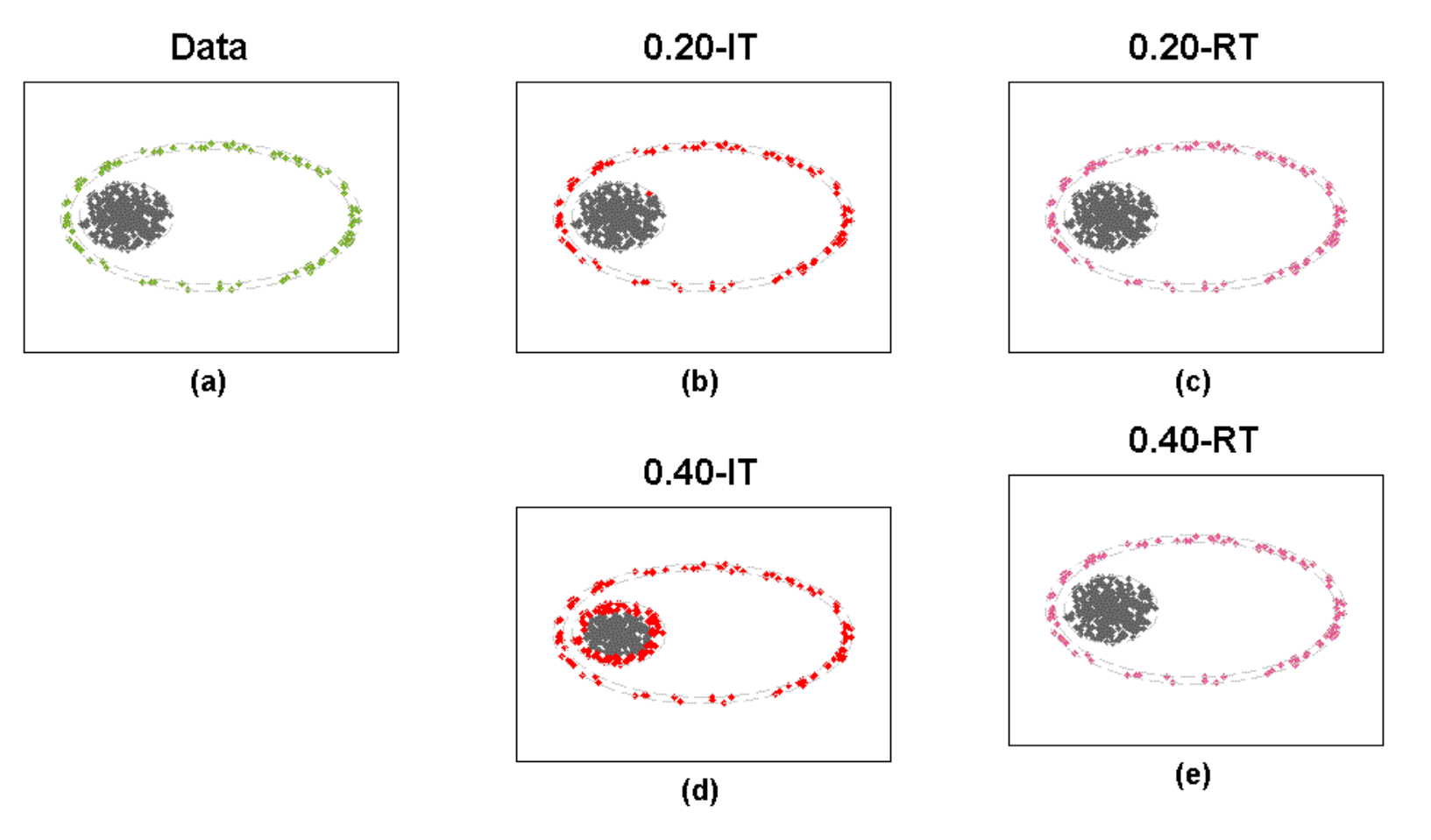}
\end{center}
\vspace{-15pt}\caption{Scatter--plot of the data with the outliers in green (a),  trimmed observations by IT $20\%$, RT $20\%$, IT $40\%$ and RT $40\%$ ((b),(c),(d) and (e)) are shown in a different color.} \label{corona}
\end{figure}

We consider samples of size $n = 100$ on dimensions $p = 10$, $20$
and $40$. In all cases a fraction $(1-\epsilon)$ of the
observations $\mathbf{x}_i$ were taken from a multivariate normal
distribution, and the remaining observations are point wise
outliers $\bf{x_0}$ .  Because of the equivariance properties of
all the estimates considered in this study, the normal
observations were taken with a mean of zero and the identity covariance
matrix. Moreover, the values of $\bf{x_0}$ were taken to have the form
$\bf{x_0}$$=\left(x_0,0,\dots,0\right)$ and $x_0$ took values of $3,
7$ and $11$. For the case of the RT-estimate, a direction was
considered to be suitable for trimming if the length of the one
dimensional gap is bigger than a constant as established on Deheuvels
(1984) we took $k=3$ and $f(x_0)=0.0044$, which corresponds to
$\varphi(3)$, where $\varphi$ is the density of a standard normal
distribution. The search is carried out in less than $100$ random
directions. In the case of the IT-estimates  we consider
$\alpha=0.5$, as suggested by Gervini (2010). In both cases the trimming
proportions were $10$, $30$ and $50\%$ of the data, for the
RT--estimates this was the maximum allowable trimming and for the
IT--estimates this is the effective pruning rate.

First, we compared the location estimates, and the results are given on
Tables \ref{centro10}  and  \ref{centro20}. In every case the RT--estimates
 outperform MVE and MVEF. The behavior of the
IT--estimate is remarkable, specially if the outliers are far away
from the bulk of the data. If $10\%$ of the observations are
outliers and the same amount of data are pruned there are several
configurations where the MSE of the trick estimate is achieved. The
RT--estimates also perform very well, in that they are better than the
IT-estimate if $x_0=3$, but not if $x_0=7$ or $11$. It is
important to note that the MSEs for the RT--estimates do not change
with the trimming rate, due to the application of a wise trimming
strategy.


\begin{table}[tbp]
\centering
\begin{tabular}{lllllllllll}
\hline
Epsilon=0.1 &  &  &  &  &  &  &  &  &  &  \\
&  & Trick & MVE & MVEF & \multicolumn{3}{l}{IT} & \multicolumn{3}{l}{RT}
\\
p & $\bf{x_0}$ &  &  &  & 0.5 & 0.7 & 0.9 & 0.5 & 0.7 & 0.9 \\ \hline
10 & 3 & 0.10 & 0.16 & 0.18 & 0.25 & 0.19 & 0.15 & 0.13 & 0.13 & 0.13 \\
& 7 & 0.10 & 0.28 & 0.19 & 0.14 & 0.12 & 0.10 & 0.23 & 0.23 & 0.23 \\
& 11 & 0.10 & 0.36 & 0.15 & 0.14 & 0.12 & 0.10 & 0.14 & 0.14 & 0.15 \\ \hline
20 & 3 & 0.10 & 0.15 & 0.17 & 0.18 & 0.15 & 0.13 & 0.11 & 0.11 & 0.12 \\
& 7 & 0.10 & 0.24 & 0.30 & 0.14 & 0.16 & 0.20 & 0.18 & 0.18 & 0.18 \\
& 11 & 0.10 & 0.35 & 0.42 & 0.14 & 0.12 & 0.10 & 0.26 & 0.26 & 0.26 \\ \hline
40 & 3 & $0.10$ & $0.13$ & $0.15$ & $0.16$ & $0.13$ & $0.11$ & $0.11$ & $0.10
$ & $0.10$ \\
& 7 & $0.10$ & $0.19$ & $0.22$ & $0.26$ & $0.20$ & $0.16$ & $0.15$ & $0.15$
& $0.15$ \\
& 11 & $0.10$ & $0.27$ & $0.33$ & $0.14$ & $0.12$ & $0.13$ &
$0.20$ & $0.20$ & $0.20$ \\ \hline
\end{tabular}%
\caption{MSE for the location estimates for $10\%$ contamination}
\label{centro10}
\end{table}

\begin{table}[tbp]
\centering
\begin{tabular}{lllllllllll}
\hline
Epsilon=0.2 &  &  &  &  &  &  &  &  &  &  \\
&  & Trick & MVE & MVEF & \multicolumn{3}{l}{IT} & \multicolumn{3}{l}{RT}
\\
p & $\bf{x_0}$ &  &  &  & 0.5 & 0.7 & 0.9 & 0.5 & 0.7 & 0.9 \\ \hline
10 & 3 & $0.11$ & $0.27$ & $0.23$ & $0.43$ & $0.32$ & $0.24$ & $0.21$ & $0.21
$ & $0.21$ \\
& 7 & $0.11$ & $0.57$ & $0.60$ & $0.14$ & $0.12$ & $0.31$ & $0.45$ & $0.45$
& $0.45$ \\
& 11 & $0.11$ & $0.88$ & $0.72$ & $0.14$ & $0.12$ & $0.40$ & $0.27$ & $0.27$
& $0.44$ \\ \hline
20 & 3 & $0.11$ & $0.23$ & $0.29$ & $0.29$ & $0.22$ & $0.18$ & $0.16$ & $0.16
$ & $0.16$ \\
& 7 & $0.11$ & $0.48$ & $0.63$ & $0.31$ & $0.45$ & $0.37$ & $0.32$ & $0.32$
& $0.32$ \\
& 11 & $0.11$ & $0.74$ & $0.99$ & $0.14$ & $0.12$ & $0.29$ & $0.49$ & $0.50$
& $0.49$ \\ \hline
40 & 3 & $0.11$ & $0.17$ & $0.21$ & $0.22$ & $0.17$ & $0.14$ & $0.13$ & $0.13
$ & $0.13$ \\
& 7 & $0.11$ & $0.33$ & $0.44$ & $0.47$ & $0.34$ & $0.27$ & $0.24$ & $0.24$
& $0.24$ \\
& 11 & $0.11$ & $0.51$ & $0.67$ & $0.14$ & $0.13$ & $0.33$ & $0.36$ & $0.36$
& $0.36$ \\ \hline
\end{tabular}%
\caption{MSE for the location estimates for $20\%$ contamination}
\label{centro20}
\end{table}

Tables \ref{correlac10} and \ref{correlac20} show the MSEs for the estimated correlation
matrices. In this case the behavior of the RT--estimates is noteworthily,
 achieving the MSE of the \emph{trick} estimate, and in almost every circumstance, this method
  outperforms the others.

\begin{table}[h]
\centering
\begin{tabular}{lllllllllll}
\hline
Epsilon=0.1 &  &  &  &  &  &  &  &  &  &  \\
&  & Trick & MVE & MVEF & \multicolumn{3}{l}{IT} & \multicolumn{3}{l}{RT}
\\
p & $\bf{x_0}$ &  &  &  & 0.5 & 0.7 & 0.9 & 0.5 & 0.7 & 0.9 \\ \hline
10 & 3 & $0.10$ & $0.12$ & $0.18$ & $0.10$ & $0.11$ & $0.14$ & $0.10$ & $0.10
$ & $0.10$ \\
& 7 & $0.10$ & $0.12$ & $0.17$ & $0.13$ & $0.11$ & $0.10$ & $0.09$ & $0.09$
& $0.09$ \\
& 11 & $0.10$ & $0.12$ & $0.17$ & $0.13$ & $0.11$ & $0.10$ & $0.10$ & $0.10$
& $0.10$ \\ \hline
20 & 3 & $0.10$ &  $0.13$ & $0.16$ & $0.14$ & $0.12$ & $0.11$ & $0.10$ & $0.10$ & $%
0.10$ \\
& 7 & $0.10$ & $0.13$ & $0.16$ & $0.15$ & $0.12$ & $0.10$ & $0.10$ & $0.10$ & $0.10
$ \\
& 11 & $0.10$ & $0.13$ & $0.16$ & $0.14$ & $0.11$ & $0.10$ & $0.10$ & $0.10$ & $%
0.10$ \\ \hline
40 & 3 & $0.10$ & $0.13$ & $0.15$ & $0.15$ & $0.13$ & $0.11$ & $0.10$ & $0.10$ & $%
0.10$ \\
& 7 & $0.10$ &  $0.13$ & $0.15$ & $0.15$ & $0.13$ & $0.11$ & $0.10$ & $0.10$ & $0.10
$ \\
& 11 & $0.10$ &  $0.13$ & $0.15$ & $0.14$ & $0.12$ & $0.10$ & $0.10$ & $0.10$ & $%
0.10$ \\ \hline
\end{tabular}%
\caption{MSE for the correlation matrix estimates for $10\%$ contamination}
\label{correlac10}
\end{table}

\begin{table}[h]
\centering
\begin{tabular}{lllllllllll}
\hline
Epsilon=0.2 &  &  &  &  &  &  &  &  &  &  \\
&  & Trick & MVE & MVEF & \multicolumn{3}{l}{IT} & \multicolumn{3}{l}{RT}
\\
p & $\bf{x_0}$ &  &  &  & 0.5 & 0.7 & 0.9 & 0.5 & 0.7 & 0.9 \\ \hline
10 & 3 & $0.11$ & $0.14$ & $0.18$ & $0.16$ & $0.12$ & $0.10$ & $0.10$ & $0.10
$ & $0.10$ \\
& 7 & $0.11$ & $0.14$ & $0.19$ & $0.13$ & $0.11$ & $0.10$ & $0.10$ & $0.10$
& $0.10$ \\
& 11 & $0.11$ &  $0.14$ & $0.19$ & $0.13$ & $0.11$ & $0.10$ & $0.10$ & $0.10$ & $%
0.10$ \\ \hline
20 & 3 & $0.11$ &   $0.15$ & $0.19$ & $0.17$ & $0.13$ & $0.11$ & $0.11$ & $0.11$ & $%
0.11$ \\
& 7 & $0.11$ &  $0.15$ & $0.19$ & $0.14$ & $0.13$ & $0.11$ & $0.11$ & $0.11$ & $0.11
$ \\
& 11 & $0.11$ &  $0.15$ & $0.19$ & $0.14$ & $0.12$ & $0.10$ & $0.11$ & $0.11$ & $%
0.10$ \\ \hline
40 & 3 & $0.11$ &  $0.14$ & $0.18$ & $0.18$ & $0.14$ & $0.12$ & $0.11$ & $0.11$ & $%
0.11$ \\
& 7 & $0.11$ &  $0.14$ & $0.18$ & $0.18$ & $0.14$ & $0.12$ & $0.11$ & $0.11$ & $0.11
$ \\
& 11 & $0.11$ &  $0.14$ &  $0.18$ & $0.14$ & $0.12$ & $0.11$ & $0.11$ & $0.11$ & $%
0.11$ \\ \hline
\end{tabular}%
\caption{MSE for the correlation matrix estimates for $20\%$ contamination}
\label{correlac20}
\end{table}
Finally, we briefly discuss the computation time. It is well
known that MVE is very slow for high dimensional data sets.  In our
simulations, the MVE is on average between 1300 and 2000 times slower that
the $\emph{trick}$ estimate. The MVEF method takes on average half
the time that MVE takes to compute the estimates; nevertheless it is a very
slow procedure. Computation times for the IT and RT estimates are very close
to the benchmark, particularly for higher dimension. For $p=40$ the two methods are
 only $30\%$ (respectively $25\%$) slower than the
trick estimate on average. For dimension $20$  they take less than twice the
time of the $\emph{trick}$ estimate to compute the estimations and
for dimension $10$ they are four times slower. In all cases the
RT--estimates are slightly faster than the IT--estimates.
\bigskip

\subsection{Functional Data}
\label{infinito} In this Section we study the performance of the
RT-estimates of the  location  and correlation functions for the
case of functional data in the presence of outliers. We generate
100 functions, 90 of which follow the central model, $
X\left( t\right) =30t(1-t)^{3/2}+\varepsilon \left( t\right)$, for $t\in \left[ 0,1\right]$.
and the remaining 10 observations are distributed under one of the
following contamination models

\textbf{Case A.}
$X\left( t\right) =30\left( 1-t\right) t^{3/2}+\varepsilon \left( t\right)$, for $ t\in \left[ 0,1\right].$ The outliers have a different shape from the functions generated by the core distribution.

\textbf{Case B.} $X\left( t\right) =30t(1-t)^{3/2}+2+\varepsilon \left( t\right)$, for $ t\in \left[ 0,1\right].$
 In these case we consider level shift outliers.

\textbf{Case C.}
\begin{equation*}
X\left( t\right) =\left\{
\begin{tabular}{lll}
$30t(1-t)^{3/2}+\varepsilon \left( t\right) $ &  & $t\in \left[ 0,0.4\right)
\cup \left( 0.6,1\right] $ \\
$30t(1-t)^{3/2}+2+\varepsilon \left( t\right) $ &  & $t\in \left[ 0.4,0.6%
\right] $%
\end{tabular}%
\right.
\end{equation*}

In the last case the outliers have the same distribution as the
observations generated under the central model, except in the
interval $[0.4,0.6]$ where they are shifted as in the previous
case. In every case the error of the process $\varepsilon \left(
t\right)$ follows an Ornstein-Uhlenbeck law, i.e., a Gaussian
process with a mean of zero and a covariance function given by,
$
\rho \left( s,t\right) =0.3\exp \left( -\frac{\left\vert s-t\right\vert }{0.3%
}\right) ,\text{ \ \ for }s,t\in \left[ 0,1\right].
$

In Figure \ref{curvascontamin} we display the three models, the grey functions
are generated under the central model while the red ones are the outliers.
\begin{figure}[ht]
\begin{center}
\includegraphics[width=8cm]{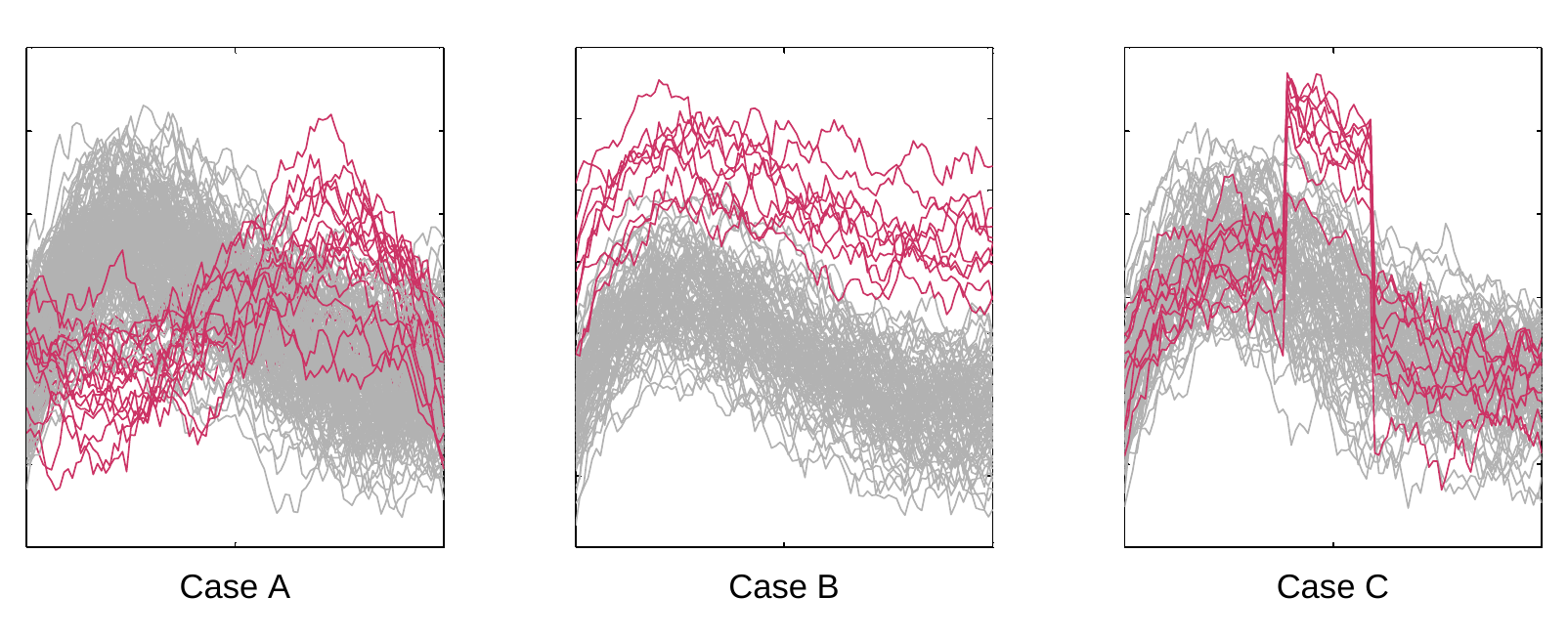}
\end{center}
\vspace{-15pt} \caption{Each figure exhibits one of the
contamination distributions.} \label{curvascontamin}
\end{figure}

\begin{figure}[tpb]
\begin{center}
\includegraphics[width=8cm]{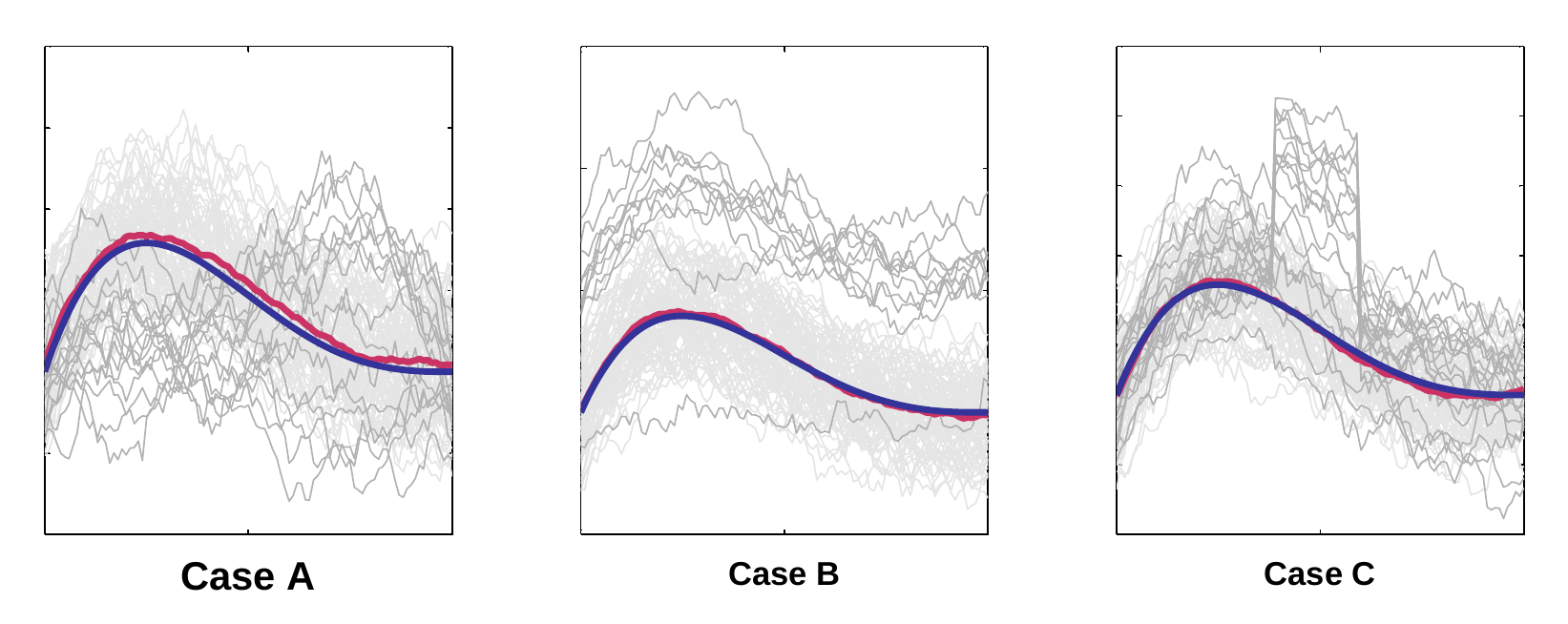}
\end{center}
\vspace{-15pt} \caption{In each figure the pruned functions are dark grey, the theoretical central curve is blue and the
estimated central curve is red.} \label{curvascentros}
\end{figure}


For each model we performed 500 replicates where a maximum of either
$20$, $30$ or $40\%$ of the observations were pruned. Each
function was observed at equidistant points on the interval $[0,1]$, and
the distance between two observations was 0.01. We
 report the mean number of outliers
pruned in table \ref{tablacanterrores}. The procedure was very successful at
detecting outliers in
Cases A and B, however these results changed significantly for Case C where
on average more than three outliers are not pruned. Moreover, in
Table \ref{tablaproppodada} we show the proportion of observations
that are effectively trimmed for each model and trimming
proportion and, it can be seen that this proportion does not change as
the trimming rate increases.

\begin{table}[tbp]
\centering
\begin{tabular}{cccc}
\hline
& \multicolumn{3}{c}{Trimmed proportion} \\
& $0.4$ & $0.3$ & $0.2$ \\ \hline
Case A & $8.08$ & $8.06$ & $8.11$ \\
Case B & $9.992$ & $9.992$ & $9.994$ \\
Case C & $6.43$ & $6.43$ & $6.42$ \\ \hline
\end{tabular}%
\caption{Mean number of outliers detected for each model and
trimming proportion.} \label{tablacanterrores}
\end{table}

\begin{table}[tbp]
\centering
\begin{tabular}{cccc}
\hline
& \multicolumn{3}{c}{Trimmed proportion} \\
& $0.4$ & $0.3$ & $0.2$ \\ \hline
Case A & $12.11$ & $12.10$ & $12.20$ \\
Case B & $13.75$ & $13.75$ & $13.84$ \\
Case C & $11.83$ & $11.88$ & $11.96$ \\ \hline
\end{tabular}%
\caption{Mean effective trimming rate for each model and trimming
proportion.} \label{tablaproppodada}
\end{table}

In order to analyze the performance of the central and correlation estimates,
 we consider the  $L_2$--error $(L_2E)$. Let
$\hat{\mu}(t)$ be the estimate of $\mu(t)$  and $\hat{\Sigma}(t,s)$
be the estimate of $\Sigma(t,s)$, where $t,s:t_1,\dots,t_N$, then
$ L_2E(\hat{\mu})=\sqrt{\frac{1}{N}\sum_{i=1}^{N} (\hat{\mu}(t_i)-\mu(t_i)) ^{2}},$

and
$L_2E(\hat{\Sigma})=\sqrt{\frac{1}{N^2}\sum_{i,j=1}^{N} (\hat{\Sigma}(t_i,s_j)-\Sigma(t_i,s_j)) ^{2}}. $

Table \ref{mseposfd} exhibits the $L_2$--error for the
location estimates. The first column shows the results for the
trick estimate: we consider these results to be the benchmark. From
the $2^{nd}$ to the $4^{th}$ columns the results for the RT--estimates
are displayed for three  different trimming levels bounds.
Finally the results for the IT--estimates proposal with
$\alpha=0.5$ and the same pruning proportion as in our case with
hard rejection weights are given on the last three columns.
The results for the RT--estimates do not change much for
the different estimates. This is because the effective trimming
rate and the number of outliers pruned in each
 replicate are practically the same.
 IT--estimates show a better performance when the trimming rate
 is lower, because the contamination proportion is even smaller than the trimming rate.
 For Models A and B our estimates outperform the IT--estimates, but in Case C the IT--estimates are better because our
 detection of outliers is less effective. Nevertheless, the performances of the two
 estimates are very similar and close to the trick estimates. Finally, Table \ref{msecorrfd}
 shows the same results for the correlation matrix estimates. The comparison among the estimates
 shows strong similarities,
 the main difference being that for Case C our estimates outperform the IT--estimates.

\begin{table}[tpb]
\centering
\begin{tabular}{cccccccc}
\hline
& Trick & \multicolumn{3}{l}{RT} & \multicolumn{3}{l}{IT} \\
&  & 0.4 & 0.3 & 0.2 & 0.4 & 0.3 & 0.2 \\ \cline{2-8}
Case A & 0.0513 & 0.0637 & 0.0640 & 0.0636 & 0.0725 & 0.0682 & 0.0640 \\
Case B & 0.0513 & 0.0595 & 0.0595 & 0.0596 & 0.0722 & 0.0677 & 0.0630 \\
Case C & 0.0513 & 0.0715 & 0.0714 & 0.0711 & 0.0705 & 0.0670 & 0.0628 \\
\hline
\end{tabular}%
\caption{Mean square errors for the location function estimate for
each model and trimming proportion for the benchmark (Trick),
RT--estimates and IT--estimates.} \label{mseposfd}
\end{table}

\begin{table}[tbp]
\centering
\begin{tabular}{cccccccc}
\hline
& Trick & \multicolumn{3}{l}{RT} & \multicolumn{3}{l}{IT} \\
&  & 0.4 & 0.3 & 0.2 & 0.4 & 0.3 & 0.2 \\ \cline{2-8}
Case A & 0.0779 & 0.1159 & 0.1160 & 0.1153 & 0.1830  & 0.1544 & 0.1262 \\
Case B & 0.0779 & 0.0957 & 0.0953 & 0.0960 & 0.1833  & 0.1543 & 0.1213 \\
Case C & 0.0779 & 0.1213 & 0.1209 & 0.1213 & 0.1814  & 0.1525 &  0.1240 \\
\hline
\end{tabular}%
\caption{Mean square errors for the correlation function estimate
for each model and trimming proportion for the benchmark (Trick),
RT-estimates and IT-estimates.} \label{msecorrfd}
\end{table}
The matlab codes for computing the RT-estimate are included as supplemental files.

\subsection{Other Classical Statistical Problems}
In this Section demonstrate the usefulness of the trimming
procedure for other classical statistical problems. The best known
dimension reduction technique is that of the principal components, the aim of which is
to find a linear combination of the original variables that are
uncorrelated and that explain as much of the variability as possible. We
consider the same models as in Section 4.2, and trim up to $40\%$
of the data.   The results are plotted in Figure \ref{grafico-
CPA}, in which it can be seen that in every case all the outliers were trimmed. It can be seen
that for the cases where the data set is complete (without
trimming) the shape of the first principal component weight
function is similar to the mean of the contamination. In the three
cases, the principal components for the trimmed data sets are
similar to the function corresponding to the data set without the
outliers.
\begin{figure}[h]
\begin{center}
\includegraphics[width=6cm]{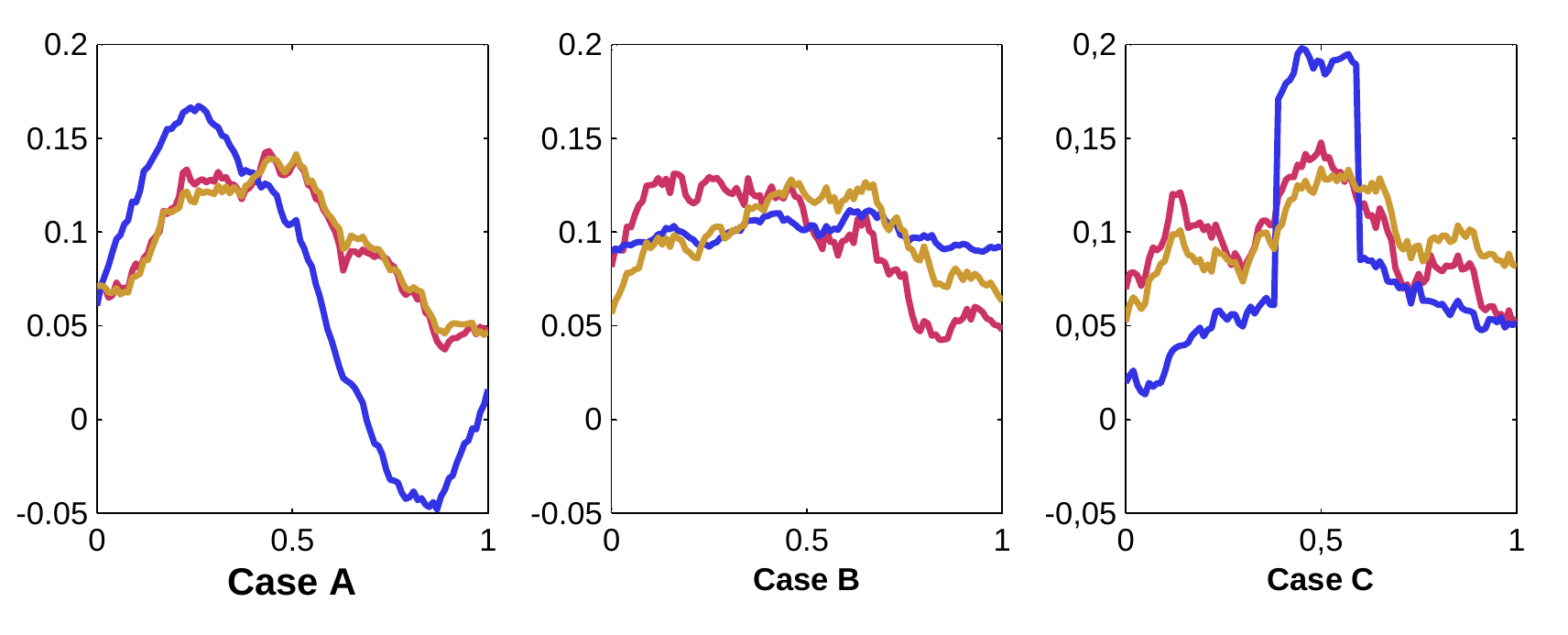}
\end{center}
\vspace{-15pt} \caption{In each figure the weight functions for
 the first principal components are represented for each contamination model.
 The blue function is for the complete data set, the yellow function is for
 the data without the outliers and the red one is for the RT data.} \label{grafico- CPA}
\end{figure}

\section{A real data example}

We now consider the data set of NOx emissions collected in Barcelona,
Spain, during the first semester of 2005. The data consists of 127
curves measured 24 times a day, and were first analyzed by
Febrero et al. (2007, 2008), in both cases they carried out procedures to detect outliers.
There are 12 days for which observations are missing, so they considered
the restricted data set conformed by 115 functions. In their first
study they found a two cluster structure, corresponding to working
days and  nonworking days (i.e., Saturdays, Sundays and holidays).
 They also provided location, scale estimates and set
confidence for the NOx data set.

Furthermore, they introduced a distance based procedure to detect
outliers. They applied the procedure to the entire set of  115 observations
and also to each of the two clusters separately.
Febrero et al. (2008) proposed  a depth-based criterion for outlier
identification, and by using  different depth notions  on
the same data set they proceeded to identify the outliers. In this case
they did not analyze the whole data set, but instead they took the subsets of working
days and non-working days separately.

The outliers detected by Febrero et al. for the non-working days
in both papers are the same (March $19^{th}$ and April $30^{th}$),
while the outliers for the working days do not strictly coincide:
in their second paper they detect two outliers (March $18^{th}$ and
April $29^{th}$) while in their first they also detect March
$11^{th}$ to be an outlier.

In this Section we apply our procedure to detect the
outliers in the data set and compare the results with those obtained by
Febrero et al. (2007, 2008).

 When trimming up to $5\%$ of the observations,
the results obtained for the working and non-working days coincide
with those obtained by Febrero et al. (2007). Figure
\ref{barpodagap} (b) and (c), shows the observations
corresponding to the working and non-working days respectively.

 In both cases the solid red observations are the outliers detected by the three procedures and the dashed red line
 is the new outlier detected only by Febrero el al. (2008).

 The analysis of the complete data set can only be compared with the results obtained on Febrero et al. (2007).
 We detect three outliers, two of which coincide with those found by Febrero el al. (2007) (March $18^{th}$ and April $29^{th}$), and
  these observations are
 shown in Figure \ref{barpodagap} (a) as a solid red line, Febrero et al. (2007)
 classify as outliers two other observations that we do
 not (March $11^{th}$ and May$2^{nd}$), which are plotted as dashed red lines and we find in
 addition an outlier that they did not, which is indicated as a solid blue line and
 corresponds to March $16^{th}$.
 The NOx level measured during the dawn and morning of March $16^{th}$ was very high, however the NOx level during the afternoon is in the bulk of the data.
But for March $11^{th}$ and May$2^{nd}$ the NOx levels during the afternoon are high compared with the others NOx levels measured at the same time, and in general these values are higher than the bulk of the data only for very short periods of time.

{\begin{figure}[h]
\begin{center}
\includegraphics[width=6cm]{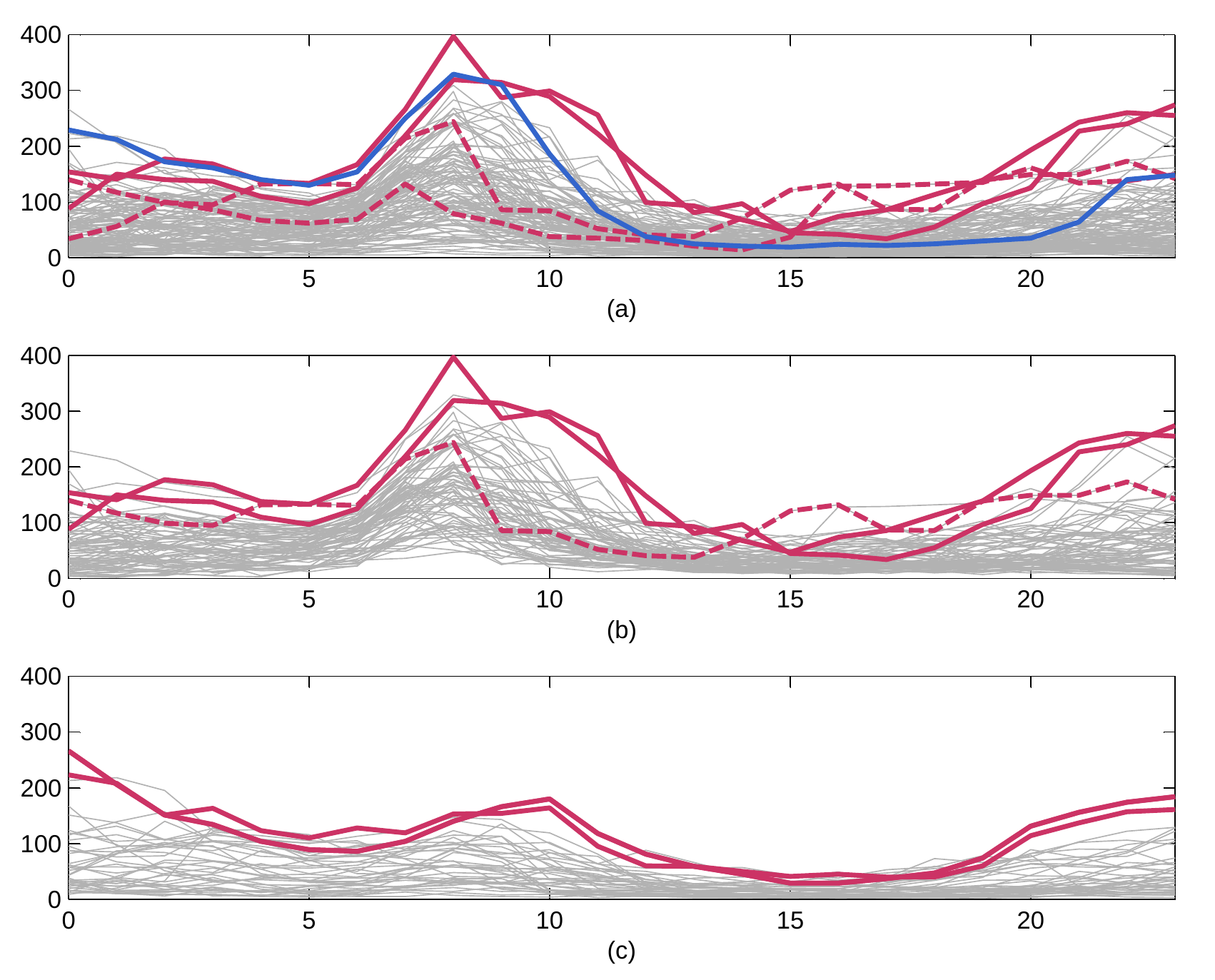}
\end{center}
\vspace{-15pt} \caption{The red lines stand for the outliers detected by Febrero et al. (2007) and by our procedure, the dashed red lines are the observations detected as outliers only by Febrero et al. (2007) and the  solid blue lines are the observations detected as outliers only by our procedure. (a) Outlier detection for the complete data. (b) Outlier detection for the working days. (c) Outlier detection for the non-working days.}
\label{barpodagap}
\end{figure}
}

\section{Concluding Remarks}

We have herein presented a new robust estimation method based on random projections that yields robust estimates of location and scatter in general Hilbert spaces. This procedure is data adaptive because it trims only the observations that are far apart from the bulk of the data.
It is a particulary suitable procedure for high dimensional data and for functional data because is these estimates are very fast to compute.

%
%
%

\end{document}